\documentclass[conference]{IEEEtran}
\IEEEoverridecommandlockouts
\usepackage{cite}
\usepackage{amsmath,amssymb,amsfonts}
\usepackage{algorithmic}
\usepackage{graphicx}
\usepackage{textcomp}
\usepackage{xcolor}
\usepackage{algorithm}
\usepackage{algorithmic}
\usepackage{amsthm}
\usepackage{amsmath,amssymb}
\usepackage{svg}
\newtheorem{problem}{Problem}

\def\BibTeX{{\rm B\kern-.05em{\sc i\kern-.025em b}\kern-.08em
    T\kern-.1667em\lower.7ex\hbox{E}\kern-.125emX}}
\begin{document}

\title{Input Matrix Optimization for Desired Reachable Set Warping of Linear Systems\\
{}
}

\author{\IEEEauthorblockN{1\textsuperscript{st} Hrishav Das}
\IEEEauthorblockA{\textit{Department of Aerospace Engineering} \\
\textit{University of Illinois Urbana-Champaign}\\
hrishav2@illinois.edu}
\and
\IEEEauthorblockN{2\textsuperscript{nd} Melkior Ornik}
\IEEEauthorblockA{\textit{Department of Aerospace Engineering} \\
\textit{University of Illinois Urbana-Champaign}\\
mornik@illinois.edu}
}

\maketitle

\begin{abstract} 
Shaping the reachable set of a dynamical system is a fundamental challenge in control design, with direct implications for both performance and safety. This paper considers the problem of selecting the optimal input matrix for a linear system that maximizes warping of the reachable set along a direction of interest. The main result establishes that under certain assumptions on the dynamics, the problem reduces to a finite number of linear optimization problems. When these assumptions are relaxed, we show heuristically that the same approach yields good results. The results are validated on two systems: a linearized ADMIRE fighter jet model and a damped oscillator with complex eigenvalues. The paper concludes with a discussion of future directions for reachable set warping research.
\end{abstract}


\section{Introduction}

In many engineering applications, a system is built around a fixed architecture, and the engineer retains only limited authority to modify its dynamics after the initial design. Two competing requirements often arise: a \emph{performance} requirement, ensuring that a capable user can drive the system to accomplish a desired task, and a \emph{safety} requirement, ensuring that an uninformed or adversarial user cannot drive it into an unsafe region. Both requirements can be framed in terms of what the system is able to reach. This paper addresses the following question for linear systems: given that the input matrix $B$ is the remaining design variable, how should it be chosen to shape the reachable set in a prescribed direction?

The reachable set of a dynamical system is the set of all states to which the system can be driven from a given initial condition within a fixed time horizon, over all admissible controls. It precisely characterizes what the system can physically achieve. Works on reachable sets can be broadly classified into three categories: (1) computation, (2) application, and (3) optimization. The current manuscript falls into the third category. Computation methods include set propagation \cite{annurev:/content/journals/10.1146/annurev-control-071420-081941}, Hamilton-Jacobi level-set methods \cite{bansal2017hamilton,mitchell2005time}, simulation-based approaches \cite{ramdani2008reachability,duggirala2013verification,fan2016locally}, and time-optimal control \cite{RACZYNSKI20231,RACZYNSKI202331} (the method used here). Reachable sets have found use in controller synthesis \cite{zamani2011symbolic,schurmann2017guaranteeing}, safe reinforcement learning \cite{gu2024review,yu2022reachability,chung2025provably}, and model predictive control \cite{bravo2006robust,gruber2019scalable}, across domains including aerospace, robotics, power systems, and systems biology \cite{annurev:/content/journals/10.1146/annurev-control-071420-081941}.

Reachable set optimization is the most nascent of the three categories, and the one this paper contributes to. The motivation is direct: expanding the reachable set along a performance-critical direction increases the range of achievable tasks, while shrinking it away from an unsafe region enforces a geometric safety constraint before any control law is designed.

We address the following problem: given a linear system whose input matrix $B$ can be selected from a compact admissible set $\mathcal{B}$, find $B^* \in \mathcal{B}$ that maximally warps the reachable set along a prescribed direction $d$. We solve this by showing, via Pontryagin's Maximization Principle (PMP), that the joint optimization over $B$ and the control input reduces to $N$ linear optimization problems, where $N$ is the number of vertices of the control polytope $\mathcal{U}$. 
\newline \indent Several prior works are related. Balashov and Kamalov \cite{balashov2023optimization} study how to choose a starting point so that the reachable set includes a target convex set in minimum time: an alternate problem, concerned with initial conditions rather than input design. Zhou and Baras \cite{7403154,7799325} shape the reachable set by varying the set of admissible control inputs to steer around collision regions; in contrast, our work fixes the control set and optimizes $B$ instead. Zhao et al.\ \cite{zhao2016scheduling} schedule control nodes in a linear time-invariant network to improve controllability, a structurally different problem. Tzoumas et al.\ \cite{tzoumas2015minimal} and Redmond and Parker \cite{redmond1996actuator} place actuators to optimize reachability under bounded input energy; the former uses a Gramian metric, which is a global, volume-like measure rather than a directional one.
\newline \indent The most closely related work is Bird et al.\ \cite{bird2024set}, which optimizes a controller to keep the closed-loop reachable set within a safe zone. Our contribution is complementary: rather than designing a controller, we select the input matrix $B$ itself to directionally warp the open-loop reachable set. The reachable set boundary is computed using time-optimal control \cite{RACZYNSKI202331} (detailed in Section~\ref{computingreachableset}).
\newline \indent The remainder of the paper is organized as follows. Section~\ref{sec:Preliminaries} presents the foundational framework for computing reachable set boundaries using PMP. Section~\ref{sec:problemstatement} formally states the problem addressed in this work. Section~\ref{sec:technicalresults} presents the main result: under two structural assumptions on the system, finding $B^*$ reduces to $N$ linear optimizations, where $N$ is the number of vertices of $\mathcal{U}$. Section~\ref{sec:examples} demonstrates the approach on two systems: a linearized ADMIRE fighter jet and a damped oscillator with complex eigenvalues, chosen to illustrate both the full-theorem regime and the behavior when assumptions are relaxed. Section~\ref{sec:Summaryandconclusion} summarizes the contributions and outlines future directions.

\section{Preliminaries}
\label{sec:Preliminaries}
In this section, we lay the foundation for the formal problem statement. We define the reachable set and explain how to compute its boundary using the Pontryagin Maximimum Principle (PMP).
\subsection{Pontryagin's Maximum Principle}
Consider optimal control of a system whose dynamics are governed by $\dot{X}(t) = f(t, X, u)$, where $X \in \mathbb{R}^n$ and $u \in \mathcal{U} \subset \mathbb{R}^d$ are the state and control, respectively, with $\mathcal{U}$ compact. We assume $f$ is such that for every admissible $u(\cdot)$ and initial condition $X_0 \in \mathbb{R}^n$, the state trajectory exists and is unique on $[0, T]$, where $T > 0$ is a fixed horizon. The objective is to minimize the cost $J = \int_0^T f_0(t, X, u)\,dt$ subject to $X(0) = X_0$. 

Pontryagin's Maximum Principle (PMP) \cite{pontryagin2018mathematical} states that if $u^*(\cdot)$ is optimal, there exists a nonzero continuous co-state $p(t) \in \mathbb{R}^n$, $t \in [0, T]$, such that the co-state dynamics are governed by
\begin{equation}
    \label{CostateDynamics}
    \Dot{P_i} = -\sum_{j=1}^n \frac{\partial \Dot{X}_j P_j}{\partial X_i} - \frac{\partial f_0}{\partial X_i},
\end{equation} where the subscript $i$ denotes the $i$-th coordinate, $\partial/\partial X_i$ is differentiation with respect to the $i$-th component of $X$, and $\dot{X}_j(t) = f_j(t, X, u)$ is the $j$-th component of the dynamics. The Hamiltonian is defined as
\begin{equation}
    \label{hamiltonion}
    \mathcal{H}(p(t), X(t), u(t)) \triangleq p(t)^\top f(t, X, u),
\end{equation}
and the optimality condition requires that $u^*(t) \in \arg\max_{u \in \mathcal{U}}\, \mathcal{H}(p(t), X(t), u)$ for all $t \in [0, T]$. Furthermore, the maximized value of $\mathcal{H}$ is constant along the optimal trajectory. We refer the reader to \cite{pontryagin2018mathematical} for the complete treatment.

\subsection{Reachable Set and its Boundary}
The reachable set of the dynamical system, denoted $\mathcal{R}(T, X_0)$, is the set of all states reachable at time $T$ from an initial condition $X(0) = X_0$ under admissible control inputs \cite{annurev:/content/journals/10.1146/annurev-control-071420-081941}. Formally,
\begin{equation}
    \mathcal{R}(T, X_0) := \left\{ X(T) \in \mathbb{R}^n \,\middle|\, 
    \begin{aligned}
        &\Dot{X}(t) = f(t, X(t), u(t)), \\
        &u(t) \in \mathcal{U} , t \in [0, T]
    \end{aligned}
    \right\}
\end{equation}

\noindent The boundary $\partial \mathcal{R}(T, X_0)$ is understood in the standard topological sense \cite{lee1967foundations}. If $\mathcal{U}$ is convex, then any trajectory passing through a boundary point at time $T$ lies entirely on the boundary for all earlier times: if $X(T) \in \partial \mathcal{R}(T, X_0)$, then $X(t) \in \partial \mathcal{R}(t, X_0)$ for all $0 \leq t \leq T$. A trajectory satisfying this property is called a \emph{boundary trajectory}. 

A key connection to optimal control is the following: let $t^* = \min\{t \geq 0 : X_f \in \mathcal{R}(t, X_0)\}$ be the minimum time needed to reach a point $X_f \in \mathbb{R}^n$. Then $X_f \in \partial \mathcal{R}(t^*, X_0)$, and the trajectory achieving this minimum is a boundary trajectory. It follows that every boundary trajectory solves a time-optimal control problem \cite{RACZYNSKI20231,RACZYNSKI202331} and can therefore be computed via PMP \cite{pontryagin2018mathematical}, as detailed in the next subsection.

\subsection{Computing the Reachable Set Boundary}
\label{computingreachableset}
 
To identify points on $\partial \mathcal{R}(T; X_0)$, which are required to evaluate $G_d(B)$, we use the following procedure based on Raczynski~\cite{RACZYNSKI202331}. Since we seek boundary trajectories, we solve a time-optimal control problem, for which by definition $f_0(t,X,u) = 1$, giving
\begin{equation}
    J = T, \quad f_0(t,X,u) = 1.
    \label{eq:general_cost}
\end{equation}
Substituting $f_0 = 1$ into \eqref{CostateDynamics} eliminates the $\partial f_0/\partial X_i$ term, yielding
\begin{equation}
    \label{CostateDynamics2}
    \Dot{P_i} = -\sum_{j=1}^n \frac{\partial \Dot{X}_j P_j}{\partial X_i}.
\end{equation}
Different solutions of \eqref{CostateDynamics2} correspond to different terminal co-state directions and thereby trace out different boundary trajectories (see Fig.~\ref{fig:dgm}). Raczynski~\cite{RACZYNSKI202331} samples boundary trajectories by drawing initial co-state values $P(0)$ from a probability distribution. In this paper, we instead parameterize boundary trajectories by their terminal co-state value: we set $P(T) = d$ for a chosen unit vector $d \in \mathbb{R}^n$, $\|d\| = 1$. Since $\mathcal{R}(T, X_0)$ stems from a single initial point, it is a convex set \cite{zampieri1994local,polyak2004convexity,Reisig2007mc}. For each unit vector $d$, the endpoint $X_d$ of the boundary trajectory with terminal condition $P(T) = d$ is a point on $\partial\mathcal{R}(T,X_0)$; sweeping over sufficiently many directions $d$ reconstructs the full boundary. The procedure to compute $\partial\mathcal{R}(T, X_0)$ is:
\begin{enumerate}
    \item Select a terminal costate value $P(T) = d \in \mathbb{R}^n$, $\|d\|=1$.
    \item Integrate \eqref{CostateDynamics2} and the system dynamics forward in time, using the control that maximizes $\mathcal{H}$ at every instant. This is a two-point boundary value problem with $X(0) = X_0$ and $P(T) = d$.
    \item Store the terminal state $X_d := X(T)$.
    \item Record $X_d$ as a point on $\partial\mathcal{R}(T, X_0)$. Return to step~1 with a new direction $d$ to collect additional boundary points; terminate once the boundary is reconstructed to the desired resolution.
\end{enumerate}

\section{Problem Statement}
\label{sec:problemstatement}
Consider the family of linear systems
\begin{equation}
    \label{lineardynamics}
    \dot{X}(t) = AX(t) + Bu(t), \quad X(0) = X_0,
\end{equation}
parameterized by $B \in \mathcal{B} \subset \mathbb{R}^{n\times d}$, where $\mathcal{B}$ is compact, $A \in \mathbb{R}^{n\times n}$ is fixed, $0 \in \mathcal{U} \subset \mathbb{R}^d$ is the compact convex control set, and $T > 0$ is the time horizon. We denote the reachable set by $\mathcal{R}_{B}(T; X_0)$.
 
Let $d \in \mathbb{R}^n$ with $\|d\| = 1$ be a direction of interest. The \emph{zero-input trajectory endpoint} is
\begin{equation}
\label{centroid}
c_0 \triangleq X_0 + \int_0^T AX(t)\, dt.
\end{equation}
Since $0 \in \mathcal{U}$, we have $c_0 \in \mathcal{R}_B(T; X_0)$. Let $X_{d,B}$ denote the endpoint of the boundary trajectory with terminal co-state $P(T) = d$, computed by the procedure of Section~\ref{computingreachableset}; the dependence on $B$ is made explicit by the subscript. These quantities are illustrated in Fig.~\ref{fig:dgm}: the blue set is $\mathcal{R}_{B}(T; X_0)$, the interior point $c_0$ is the zero-input endpoint, and the red arrow at $X_{d,B}$ shows the terminal co-state direction $d$.
 
We define a quantity called the \textit{Directional Growth Metric} as follows:
\begin{equation}
\label{DGrowthMetric}
G_d(B) \triangleq d^\top\!\bigl(X_{d,B}-c_0\bigr).
\end{equation} Since $d$ and $c_0$ are independent of $B$, the influence of $B$ on $G_d$ is entirely through $X_{d,B}$. An increase in $G_d(B)$ indicates that the reachable set has expanded in the direction $d$. Note that $G_d$ is continuous in $B$ (since $X_{d,B}$ depends continuously on $B$ via the ODE flow), and since $\mathcal{B}$ is compact, there exists a maximum. We therefore pose the following problem formally:

\begin{problem}
\label{prob:main}
Given the system \eqref{lineardynamics} with $A$ fixed, $B \in \mathcal{B}$, direction $d$, and horizon $T$, find
\begin{equation}
    B^{\ast} \in \arg\max_{B \in \mathcal{B}}\, G_d(B).
\end{equation}
\end{problem}

\begin{figure}
    \centering
    \includegraphics[width=1\linewidth]{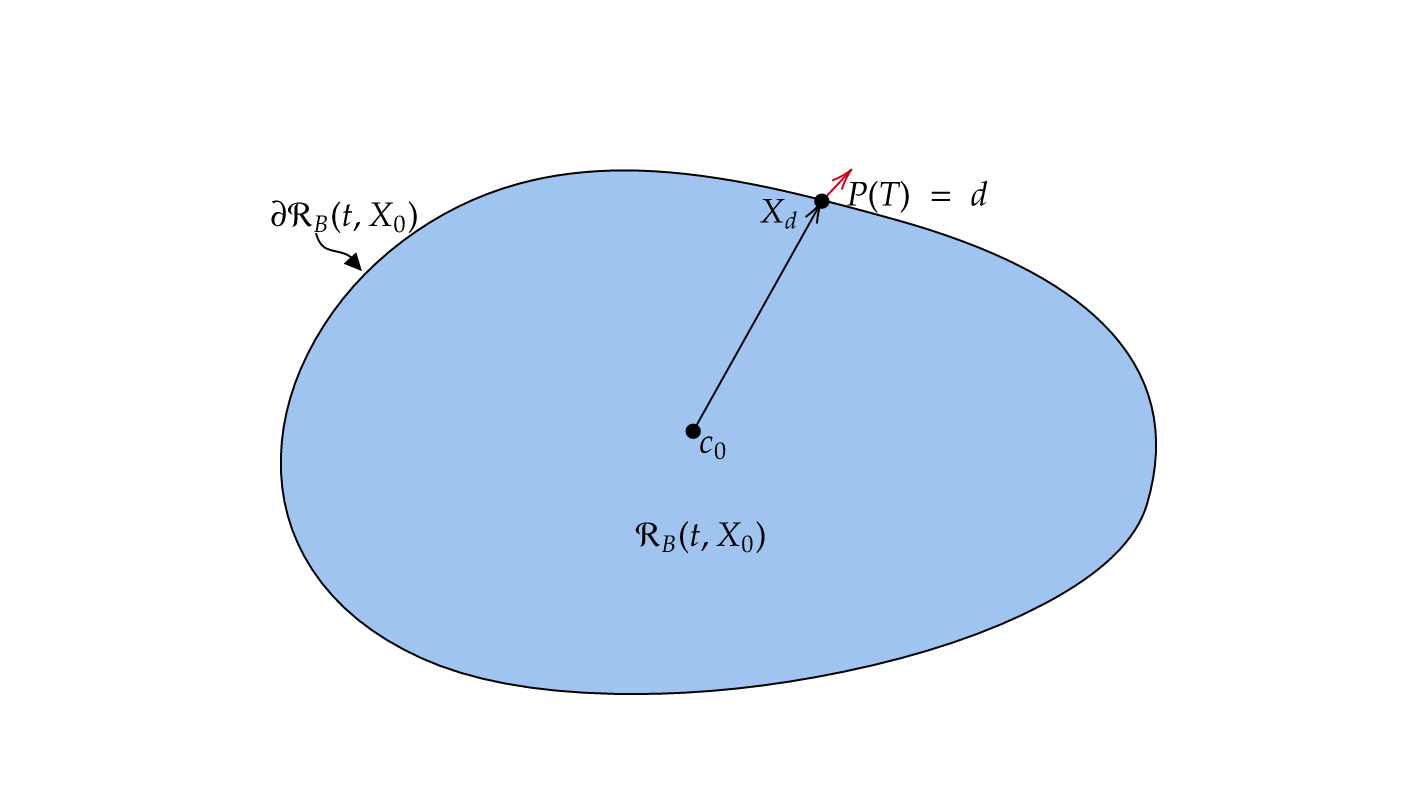}
    \caption{Parameters of the Directional Growth Metric $G_d(B) = d^\top(X_{d,B} - c_0)$. The blue region is the reachable set $\mathcal{R}_B(T; X_0)$, $c_0$ is the zero-input trajectory endpoint (interior point), $X_{d,B}$ is the boundary point reached by the trajectory with terminal co-state $P(T) = d$, and the red arrow indicates the direction $d$.}
    \label{fig:dgm}
\end{figure}

\section{Technical Results}
\label{sec:technicalresults}
We now state and prove the main result: under two structural assumptions, Problem~\ref{prob:main} reduces to $N$ linear optimizations (one per vertex of $\mathcal{U}$), and the global maximizer $B^\ast$ is found by a single comparison across $N$ candidates. The result requires:
\begin{enumerate}
    \item $A$ has real eigenvalues,
    \item $d \in \mathbb{R}^n$, $\|d\|=1$, is an eigenvector of $A^\top$.
\end{enumerate}
 
\newtheorem{theorem}{Theorem}
\newtheorem{definition}{Definition}
 
\begin{theorem}
\label{thm:directional_optimal_B}
Consider the family of systems \eqref{lineardynamics} parameterized by $B \in \mathcal{B}$, where $\mathcal{U} \subset \mathbb{R}^d$ is a compact convex polytope (i.e., the convex hull of finitely many vertices $\{u_i\}_{i=1}^N$) and $T > 0$. Under Assumptions~1--2 of Problem~\ref{prob:main}, define $P_0 \triangleq e^{A^\top T} d \in \mathbb{R}^n$. For each vertex $u_i$ of $\mathcal{U}$, let
\[
B_i \triangleq \arg\max_{B\in\mathcal{B}}\, P_0^\top B u_i, \qquad
i^\ast \triangleq \arg\max_{1\le i\le N}\, P_0^\top B_i u_i.
\]
Then $B^\ast \triangleq B_{i^\ast}$ solves Problem~\ref{prob:main}, i.e., $B^\ast \in \arg\max_{B\in\mathcal{B}}\, G_d(B)$.
\end{theorem}
 
\begin{proof}
We reduce the joint $(B,u)$ maximization of the Hamiltonian to a finite candidate set $\{B_i\}$ determined by the vertices of $\mathcal{U}$, then show that maximizing the Hamiltonian implies maximizing $G_d(B)$.
 
\medskip
\noindent\textbf{Step 1: PMP for the linear system and vertex structure of the optimal control.}
For the linear system
\[
\dot{X}(t)=AX(t)+Bu(t), \qquad X(0)=X_0,
\]
with $u(t)\in\mathcal{U}$, where $\mathcal{U}$ is a convex polytope,
fix a direction of interest $d\in\mathbb{R}^n$ with $\|d\|=1$ and impose the terminal co-state condition $P(T)=d$ (see Section~\ref{computingreachableset}), since we seek the boundary point $X_d(T)$. By PMP, the costate \eqref{CostateDynamics2} satisfies 
\begin{equation}
\dot{P}(t) = -A^\top P(t),
\qquad
P(T)=d,
\label{eq:costate_linear}
\end{equation}
whose solution is
\begin{equation}
P(t)=\exp\!\big(-A^\top(t-T)\big)\,d.
\label{eq:Pt_solution}
\end{equation}
The Hamiltonian is given by
\begin{equation}
\mathcal{H}(t) = \max_{u(t)\in\mathcal{U}} \Big(P(t)^\top AX(t) + P(t)^\top B u(t)\Big).
\label{eq:H_linear}
\end{equation}
For a given $A,B$ and $X(t)$, the term $P(t)^\top AX(t)$ is independent of $u$, hence the
optimal control satisfies the pointwise maximization
\begin{equation}
u_B^\ast(t)\in\arg\max_{u\in\mathcal{U}} P(t)^\top Bu.
\label{eq:Ustar_vertex}
\end{equation}
Since $\mathcal{U}$ is a convex polytope and $u \mapsto P(t)^\top Bu$ is linear in $u$, the maximum over $\mathcal{U}$ is attained at one of its vertices. Let $\{u_i\}_{i=1}^N$ denote the vertices of
$\mathcal{U}$. Then, for each $t$,
\begin{equation}
\max_{u\in\mathcal{U}} P(t)^\top Bu = \max_{1\le i\le N} P(t)^\top B u_i.
\label{eq:vertex_reduction}
\end{equation}
 
\medskip
\noindent\textbf{Step 2: Hamiltonian maximization over $B$.}
A direct attempt to optimize $\mathcal{H}(t)$ jointly over $u(t)\in\mathcal{U}$ and $B\in\mathcal{B}$
at every time,
\[
\mathcal{H}(t)=\max_{u(t)\in\mathcal{U},\,B\in\mathcal{B}}
\Big(P(t)^\top AX(t) + P(t)^\top Bu(t)\Big),
\]
would in general produce a time-varying $B(t)$, which is not admissible since $B$ is a fixed system parameter. Instead, we use the fact that $\mathcal{H}(t)$ is constant along any boundary trajectory \cite{pontryagin2018mathematical}. Therefore, maximizing $\mathcal{H}$ over $B \in \mathcal{B}$ is equivalent to maximizing it at any single time; we choose $t = 0$:
Fixing $t = 0$, the Hamiltonian (evaluated at the optimal control) is
\begin{equation}
\mathcal{H}(0) = P(0)^\top A X(0) + \max_{u \in \mathcal{U}}\, P(0)^\top B u,
\label{eq:H0_tilde}
\end{equation}
where the term $P(0)^\top A X(0)$ is independent of $B$, so maximizing $\mathcal{H}(0)$ over $B \in \mathcal{B}$ reduces to maximizing $\max_{u \in \mathcal{U}} P(0)^\top B u$ over $B$.

\noindent Define
\begin{equation}
P_0 \triangleq P(0) = \exp\!\big(A^\top T\big)\,d,
\label{eq:P0_def_slides}
\end{equation}
which follows from \eqref{eq:Pt_solution} evaluated at $t=0$. Substituting $P(0) = P_0$ into \eqref{eq:H0_tilde} and using the vertex property \eqref{eq:vertex_reduction}, the maximization over $B$ and $u$ becomes
\begin{equation}
\max_{B\in\mathcal{B},\, u\in\mathcal{U}} P_0^\top B u
= \max_{B\in\mathcal{B}} \max_{1 \le i \le N} P_0^\top B u_i.
\label{eq:H0_substituted}
\end{equation}
Moreover, by the vertex property in \eqref{eq:vertex_reduction}, for any $B$ the maximizing control
$u^*_B(0)$ must be one of the vertices $\{u_i\}_{i=1}^N$.
So, we evaluate the maximization across all vertices.

\medskip
\noindent\textbf{Step 3: Finite candidate set $\{B_i\}$ and global selection.}
For each vertex $u_i$ of $\mathcal{U}$, define the candidate matrix
\begin{equation}
B_i \triangleq \arg\max_{B\in\mathcal{B}} P_0^\top B u_i.
\label{eq:Bi_def}
\end{equation}
Then select the best vertex--matrix pair via
\begin{equation}
i^\ast \triangleq \arg\max_{1\le i\le N} \; P_0^\top B_i u_i,
\label{eq:i_star_def}
\end{equation}
and set $B^\ast \triangleq B_{i^\ast}$. By construction, the pair $(B^\ast, u_{i^\ast})$
achieves the global maximum of the initial-time Hamiltonian over $B\in\mathcal{B}$ and $u\in\mathcal{U}$, and thus yields the globally maximal Hamiltonian level among admissible constant matrices $\mathcal{B}$.

\medskip
\noindent\textbf{Step 4: Maximizing the Hamiltonian implies maximizing $G_d(B)$.}
We now show that the $B^\ast$ found in Step~3 also maximizes the directional growth metric $G_d(B)$. The key observation is that, under Assumption~2, the co-state $P(t)$ remains proportional to $d$ for all $t \in [0,T]$. Specifically, from \eqref{eq:Pt_solution}, since $d$ is an eigenvector of $A^\top$, $P(t)$ remains collinear with $d$ for all $t \in [0,T]$. Hence $P(t) = \lambda(t)\,d$ for some scalar-valued function $\lambda : [0,T] \to \mathbb{R}$, where
\begin{equation}
\lambda(t) \triangleq \frac{P(t) \cdot d}{\|d\|^2} = P(t) \cdot d \quad (\text{since } \|d\|=1).
\label{eq:P_lambda_d}
\end{equation}
Under Assumption~1 ($A^\top$ has only real eigenvalues), $\lambda(t)$ remains real for all $t \in [0,T]$. Since $P(T) = d$ implies $\lambda(T) = 1 > 0$, and $\lambda(t) = e^{\mu(T-t)}$ where $\mu$ is the real eigenvalue of $A^\top$ associated with $d$ \cite{pontryagin2018mathematical}, $\lambda(t) > 0$ for all $t \in [0,T]$.

\noindent For the linear system, along the optimal trajectory (where $u = u^*_B(t)$), the Hamiltonian evaluates to
\begin{equation}
\mathcal{H}(t) = P(t)^\top \dot{X}(t),
\label{eq:H_equals_PdotX}
\end{equation}
since $P(t)^\top A X(t) + P(t)^\top B u^*_B(t) = P(t)^\top (AX(t) + Bu^*_B(t)) = P(t)^\top \dot{X}(t)$.
Substituting \eqref{eq:P_lambda_d} into \eqref{eq:H_equals_PdotX} gives
\begin{equation}
\mathcal{H}(t) = \lambda(t)\, d^\top \dot{X}(t).
\label{eq:H_lambda_d}
\end{equation}
Since $\lambda(t) > 0$ for all $t \in [0,T]$, maximizing $\mathcal{H}(t)$ over $B \in \mathcal{B}$ is equivalent to maximizing $d^\top \dot{X}(t)$, i.e.,
\[
\max_{B\in\mathcal{B}} \mathcal{H}(t)
\;\Longleftrightarrow\;
\max_{B\in\mathcal{B}} \lambda(t)\,d^\top\dot{X}(t)
\;\Longleftrightarrow\;
\max_{B\in\mathcal{B}} d^\top\dot{X}(t).
\]

Integrating both sides over $t \in [0,T]$ and noting that maximizing $d^\top \dot{X}(t)$ pointwise over $B$ implies maximizing its time integral (since all operations are linear in $B$):
\begin{align}
\max_{B\in\mathcal{B}} d^\top \dot{X}(t)
&\implies \max_{B\in\mathcal{B}} \int_0^T d^\top \dot{X}(t)\,dt \nonumber\\
&= \max_{B\in\mathcal{B}} d^\top \int_0^T \dot{X}(t)\,dt \nonumber\\
&= \max_{B\in\mathcal{B}} d^\top\!\big(X_{d,B}(T)-X_0\big) \label{eq:integrate_to_G}\\
&= \max_{B\in\mathcal{B}} d^\top\!\big(X_{d,B}(T)-c_0\big) \nonumber\\
&= \max_{B\in\mathcal{B}} G_d(B). \nonumber
\end{align}
Here, the third line uses the fundamental theorem of calculus, and the fourth uses that $c_0 = X_0 + \int_0^T AX(t)\,dt$ equals $X_0$ plus a $B$-independent term, so $X_0$ and $c_0$ differ by a constant. Therefore, the $B^\ast = B_{i^\ast}$ that maximizes $\mathcal{H}(0)$ also maximizes $G_d(B)$.
\end{proof}

\noindent The proof is constructive: given $\mathcal{U}$ in vertex representation (i.e., as the convex hull of $\{u_i\}_{i=1}^N$, which is how polytopes are typically stored \cite{ziegler1995lectures}), $B^\ast$ is recovered by solving $N$ linear optimization problems of the form $\max_{B\in\mathcal{B}} P_0^\top B u_i$ and selecting the maximizing pair. This is summarized in Algorithm~\ref{alg:iterative_B_update} for completeness. Note that the algorithm as stated applies to the class of linear systems satisfying Assumptions~1 and~2.
\begin{algorithm}[t]
\caption{Optimal Input Matrix $B^\ast$}
\label{alg:iterative_B_update}
\begin{algorithmic}[1]
\STATE \textbf{Given:} Linear system $\dot{x}(t)=Ax(t)+Bu(t)$, control set $\mathcal{U}$ (convex polytope with vertices $\{u_i\}_{i=1}^N$), admissible matrix set $\mathcal{B}$ (compact), terminal direction $d$ (eigenvector of $A^\top$, $\|d\|=1$), horizon $T > 0$, and $A$ with real eigenvalues. \textit{(Assumptions~1 and~2 of Theorem~\ref{thm:directional_optimal_B}.)}

\STATE Compute $P_0 \gets \exp(A^\top T)\, d$

\STATE Determine the set of vertices $\{U_i\}_{i=1}^N$ of $\mathcal{U}$

\FOR{$i = 1$ to $N$}
    \STATE Compute
    \[
        B_i \gets \arg\max_{B \in \mathcal{B}} P_0^\top B U_i
    \]
\ENDFOR

\STATE Select optimal index
\[
    i^\ast \gets \arg\max_{1 \le i \le N} P_0^\top B_i U_i
\]

\STATE \textbf{Return:} $B^\ast \gets B_{i^\ast}$
\end{algorithmic}
\end{algorithm}

\subsection{Relaxing Assumptions of Theorem~\ref{thm:directional_optimal_B}}

This subsection discusses what happens when Assumptions 1 and 2 are relaxed. Rather than formal guarantees of global optimality, we present a heuristic explanation for why Algorithm~\ref{alg:iterative_B_update} continues to produce meaningful results in both cases.

\subsubsection{Complex eigenvalues of $A^\top$ (both Assumptions relaxed)}
\label{complexeigdiscussion}

When $A^\top$ has eigenvalues with nonzero imaginary part, the output $B^\ast$
of Algorithm~\ref{alg:iterative_B_update} expands $\mathcal{R}_{B^\ast}(T; X_0)$
in multiple directions simultaneously, not only along $d$.

The heuristic reasoning is as follows. From \eqref{eq:Pt_solution},
\[
P(t) = e^{-A^\top(t-T)}\,d,
\]
and when $A^\top$ has complex eigenvalues, this matrix exponential involves
oscillatory modes. $P(t)$ therefore rotates through multiple directions over
$[0,T]$ rather than remaining proportional to $d$. At each time, the optimal
control is
\[
u^\ast(t) \in \arg\max_{u \in \mathcal{U}}\, P(t)^\top B u,
\]
so the control tracks $P(t)$ as it sweeps through these directions. By the
integration argument of Step~4, growth is induced along every direction that
$P(t)$ visits. The breadth of this multi-directional expansion scales with
$|\operatorname{Im}(\lambda)|$: larger imaginary parts cause $P(t)$ to sweep a
wider arc over $[0,T]$. Crucially, the terminal condition still enforces $P(T) = d$, so among all the
directions swept by $P(t)$, the direction $d$ is always one of them. This means
the algorithm remains biased to include $d$ at the boundary point $X_{d,B}(T)$,
even as it simultaneously drives growth in other directions. The result is
therefore a near-good solution for $G_d$: not globally optimal as in
Theorem~\ref{thm:directional_optimal_B}, but directionally informed rather than
arbitrary.

\subsubsection{Real eigenvalues of $A^\top$ but $d$ not an eigenvector (Assumption~2 relaxed)}
\label{realeigendiscussion}

A similar argument holds for the case in which only Assumption 1 holds. When $d$ is not an eigenvector of $A^\top$, the scalar decomposition $P(t) = \lambda(t)\,d$ in \eqref{eq:P_lambda_d} no longer holds globally, so the equivalence between maximizing $\mathcal{H}(t)$ and maximizing $d^\top \dot{X}(t)$ point-wise breaks down. Consequently, the global optimality argument of Step~4 in the proof no longer applies, and $B^\ast$ returned by Algorithm~\ref{alg:iterative_B_update} is not guaranteed to be the global maximizer of $G_d$.

Even when $P(t) \neq \lambda(t)\,d$ for intermediate times, the terminal condition $P(T) = d$ anchors the co-state to the direction of interest at terminal time (the boundary point). $X_{d,B}(T)$ is determined by the trajectory endpoint, and the terminal co-state is always $d$ regardless of $B$. The algorithm therefore optimizes $\mathcal{H}$ along a trajectory that is \emph{guaranteed to terminate in the direction $d$}, which is why growth along $d$ is assured even without Assumption~2.

\section{Examples}
\label{sec:examples}

We illustrate Theorem~\ref{thm:directional_optimal_B} on two systems. The first (Section~\ref{sec:admire}) has all real eigenvalues and demonstrates controlled directional growth and shrinkage. The second (Section~\ref{sec:complexeig}) has complex eigenvalues, violating Assumption~1, and illustrates the multi-directional growth predicted by Section~\ref{realeigendiscussion}.

\subsection{ADMIRE Linearized Fighter Jet Model}
\label{sec:admire}
We demonstrate Theorem~\ref{thm:directional_optimal_B} on the linearized ADMIRE fighter jet model,
restricting attention to the states directly driven by the control inputs. The states are the roll,
pitch, and yaw rates, denoted $p$, $q$, and $r$ (rad/s), and the four inputs correspond
to deflections (in radians) of the canard wings, left and right elevons, and the rudder.
A linearized model was established in \cite{HARKEGARD2005137} as $\dot{x} = Ax + Bu$, where
$$
x =  
\begin{bmatrix}
p\\q\\r
\end{bmatrix}, \quad A =
\begin{bmatrix}
-0.9967 & 0       &  0.6176 \\
 0       & -0.5057 &  0      \\
-0.0939  & 0       & -0.2127
\end{bmatrix},
$$
$$
B =
\begin{bmatrix}
0       & -4.2423 &  4.2423 &  1.4871 \\
1.6532  & -1.2735 & -1.2735 &  0.0024 \\
0       & -0.2805 &  0.2805 & -0.8820
\end{bmatrix}.
$$
The control set $\mathcal{U}$ imposes symmetric barrier constraints of $\pm 0.1$ radians ($\approx \pm 5.7^\circ$) on each control surface, consistent with small-perturbation linearization validity. The
admissible set for $B$ is taken as $\mathcal{B} = \left\{ M \in \mathbb{R}^{3\times4} \;:\; \|M-B_{\mathrm{nom}}\|_F \le 0.5 \right\}$, a Frobenius norm ball of radius $0.5$ around the
nominal $B$. The time horizon is $T = 2$ seconds, chosen to allow several oscillation periods while remaining within the linear regime.

We first consider a design scenario in which a broader range of an expanded range of
roll rate $p$ is required. Note that the eigenvalues of $A$ are approximately $-0.997$, $-0.506$, and $-0.213$ (all real), and $d = [1,0,0]^\top$ is an eigenvector of $A^\top$, so both assumptions of Theorem~\ref{thm:directional_optimal_B} are satisfied. Running Algorithm~\ref{alg:iterative_B_update} for growth yields Fig.~\ref{fig:admirejetresult_P} (red: nominal $\mathcal{R}_B$, green: optimized $\mathcal{R}_{B^\ast}$). The reachable set has expanded along $d$ as guaranteed. Shrinkage along $p$ is obtained by replacing $\arg\max$ with $\arg\min$ in the optimization; the result is shown in Fig.~\ref{fig:admirejetresult_Pminus} (same color convention). This corresponds to minimizing $G_d(B)$ rather than maximizing it; Theorem~\ref{thm:directional_optimal_B} applies symmetrically to this case.

Finally, we consider growth along $d = [0.3536,\,0.6124,\,0.7071]^\top$, a direction with approximately equal pitch and yaw weighting (and half the roll weight), illustrating that the algorithm handles arbitrary eigenvector directions. The result is shown in Fig.~\ref{fig:admirejetresult}.
\begin{figure}[h!]
    \centering
    \includegraphics[width=1\linewidth]{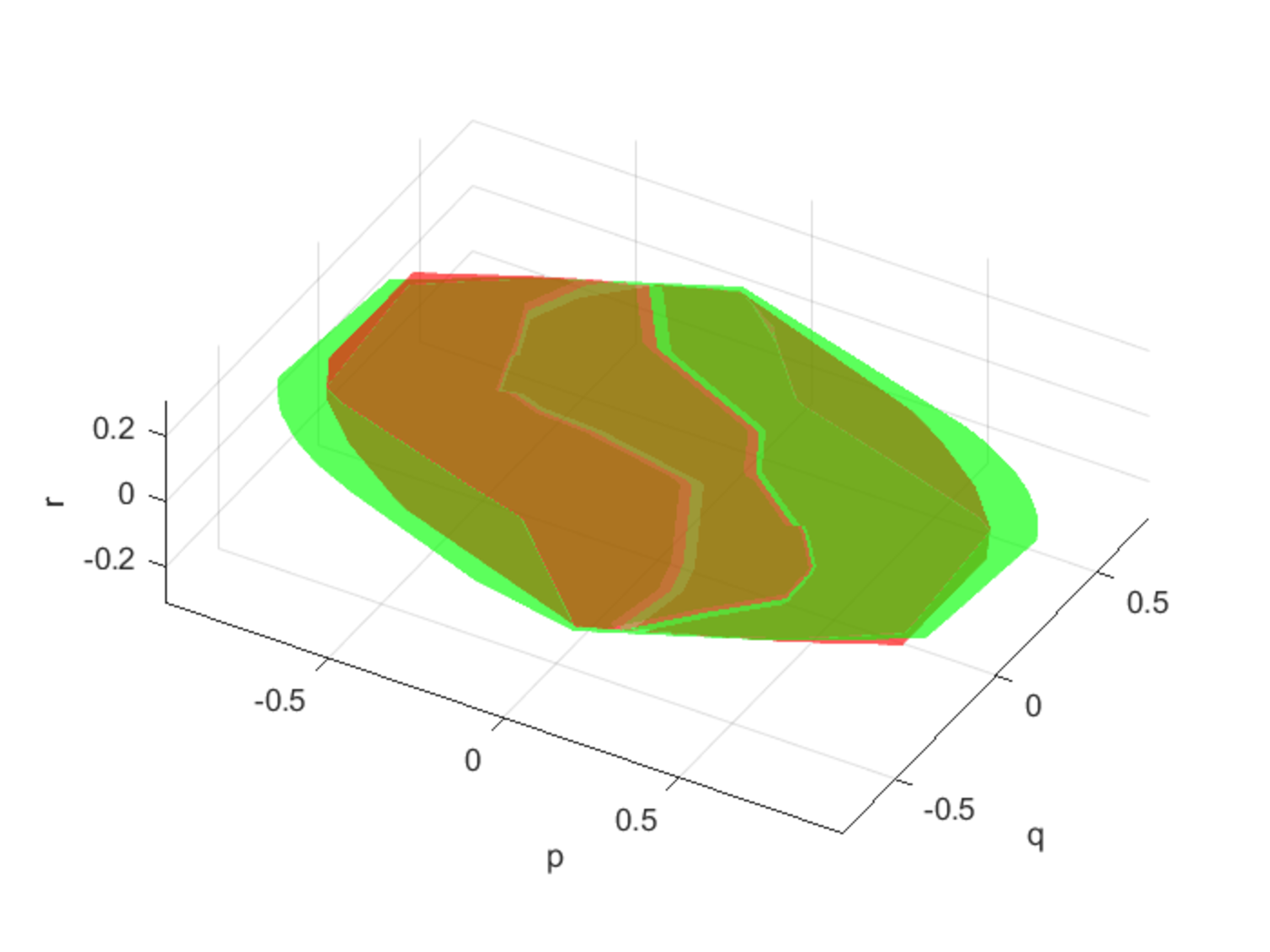}
    \caption{Growth along $p$ direction ($d = [1,0,0]^\top$). Red: nominal $\mathcal{R}_B$. Green: optimized $\mathcal{R}_{B^\ast}$.}
    \label{fig:admirejetresult_P}
\end{figure}

\begin{figure}[h!]
    \centering
    \includegraphics[width=1\linewidth]{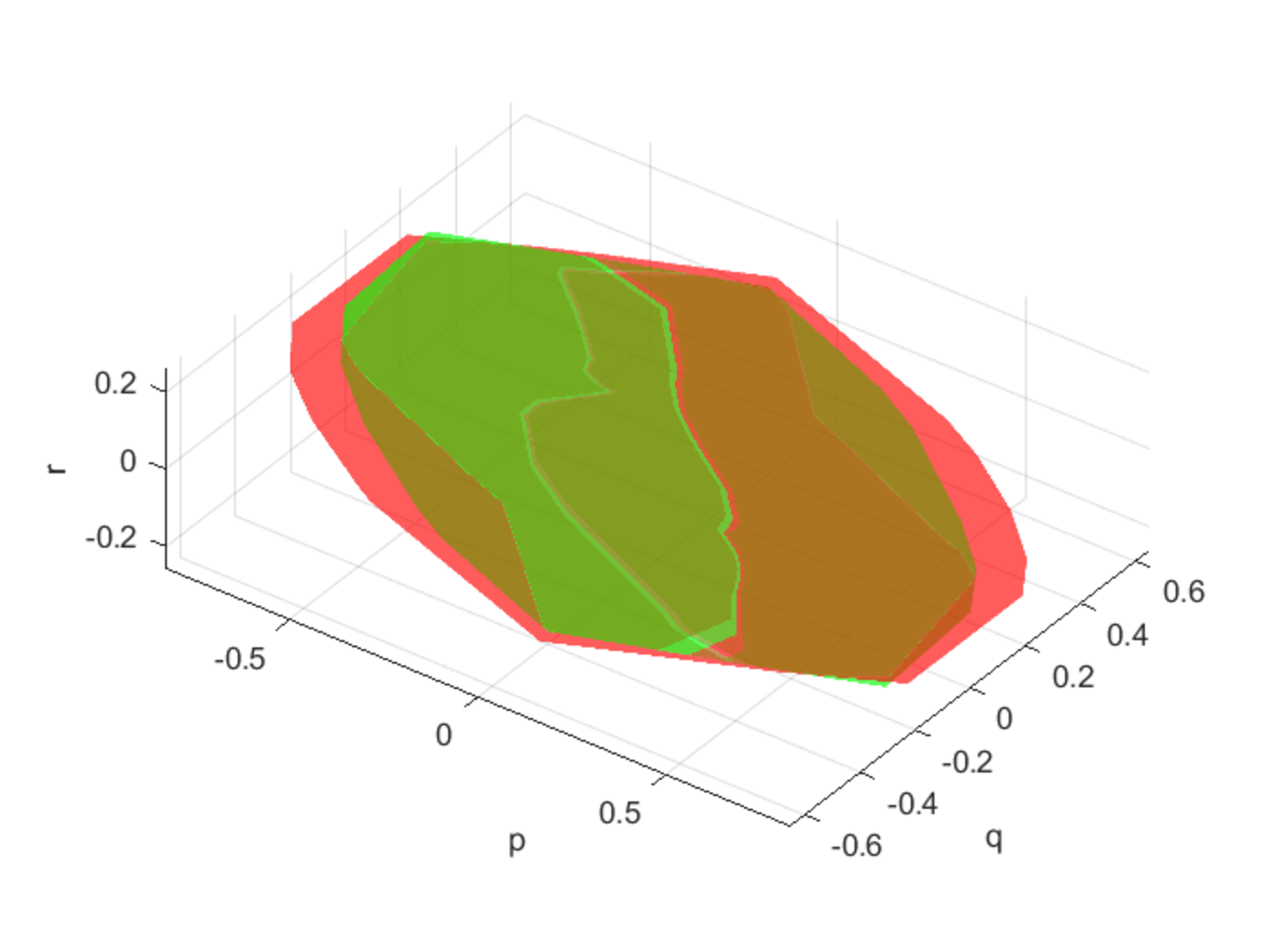}
    \caption{Shrinkage along $p$ direction ($d = [1,0,0]^\top$, $\arg\min$ variant). Red: nominal $\mathcal{R}_B$. Green: optimized $\mathcal{R}_{B^\ast}$.}
    \label{fig:admirejetresult_Pminus}
\end{figure}

\begin{figure}[h!]
    \centering
    \includegraphics[width=0.9\linewidth]{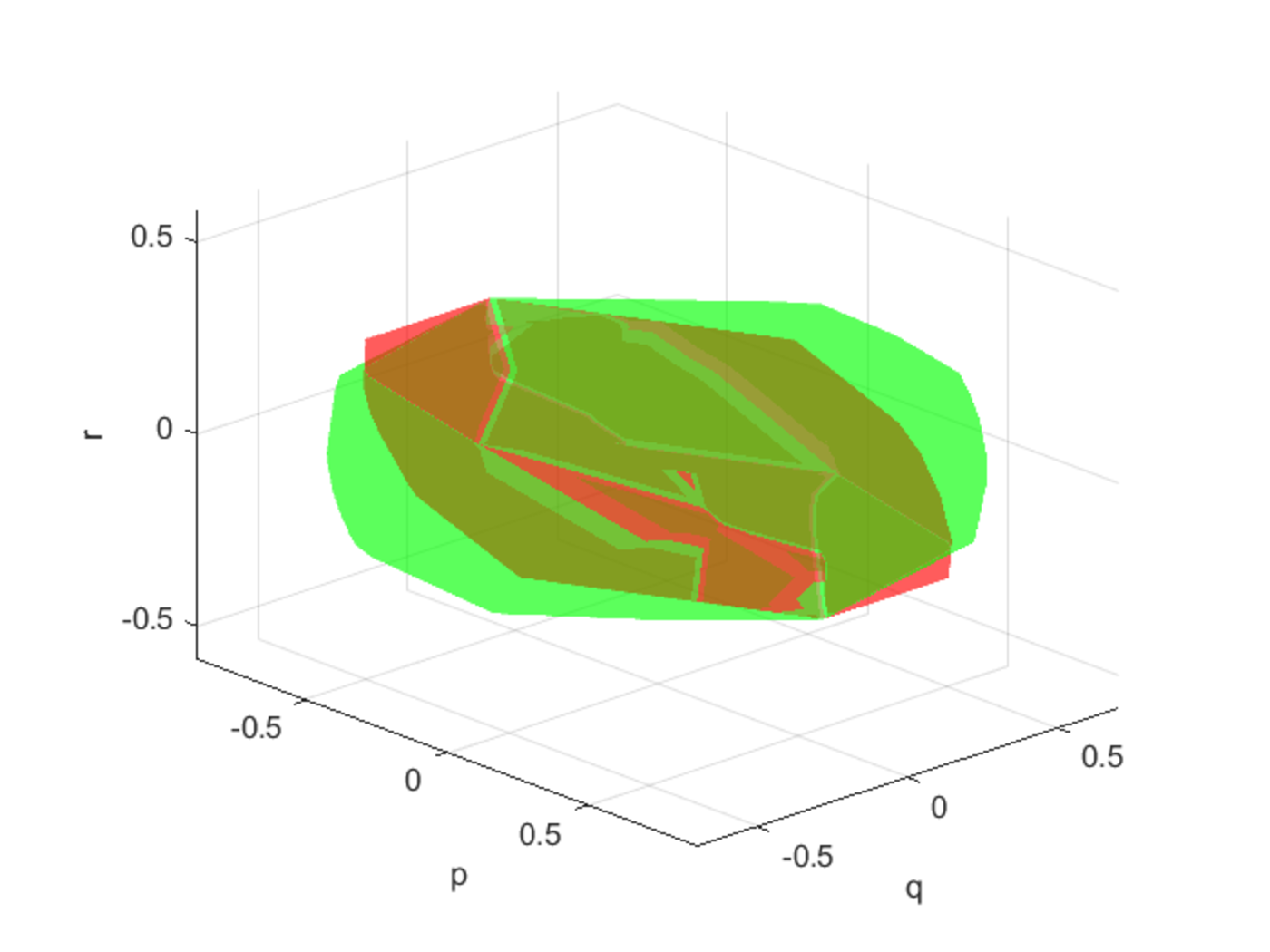}
    \caption{Growth along $d = [0.3536,\,0.6124,\,0.7071]^\top$ (approximately equal pitch and yaw weighting). Red: nominal $\mathcal{R}_B$. Green: optimized $\mathcal{R}_{B^\ast}$.}
    \label{fig:admirejetresult}
\end{figure}

\subsection{System with Imaginary Eigenvalues of $A$}
\label{sec:complexeig}
We next examine a stable damped oscillator. The system matrix $A$ has complex eigenvalues, so Theorem~\ref{thm:directional_optimal_B} does not apply and optimality of $B^\ast$ is not guaranteed; what is guaranteed, per Proposition~\ref{complexeigdiscussion}, is multi-directional growth. The two inputs act independently on velocity and position:
$$
A = \begin{bmatrix} 0 & 1 \\ -2 & -0.8 \end{bmatrix}, \quad
B_{\mathrm{nom}} = \begin{bmatrix} 0 & 1 \\ 1 & 0 \end{bmatrix},
\quad x \in \mathbb{R}^2.
$$
The control set is $\mathcal{U} = [-1,1]^2$. The admissible set is $\mathcal{B} = \left\{ M \in \mathbb{R}^{2\times2} \;:\; \|M - B_{\mathrm{nom}}\|_F \le 0.5 \right\}$. The time horizon is $T = 2$ seconds.
With $d = [1,0]^\top$, Algorithm~\ref{alg:iterative_B_update} is run for growth on
$\mathcal{R}_{B}(2;\,[0,0]^\top)$. The result is shown in Fig.~\ref{fig:imaginaryeigenvalues} (red: nominal $\mathcal{R}_{B_{\mathrm{nom}}}$, green: optimized $\mathcal{R}_{B^\ast}$). The reachable set appears to be a polytope because it is sampled at sparse points. As established in Section~\ref{complexeigdiscussion}, the reachable set expands in all directions rather than solely along $d$, because $P(t) = e^{-A^\top(t-T)}d$ oscillates through multiple directions over $[0,T]$ when $A^\top$ has complex eigenvalues.
\begin{figure}[h]
    \centering
    \includegraphics[width=1\linewidth]{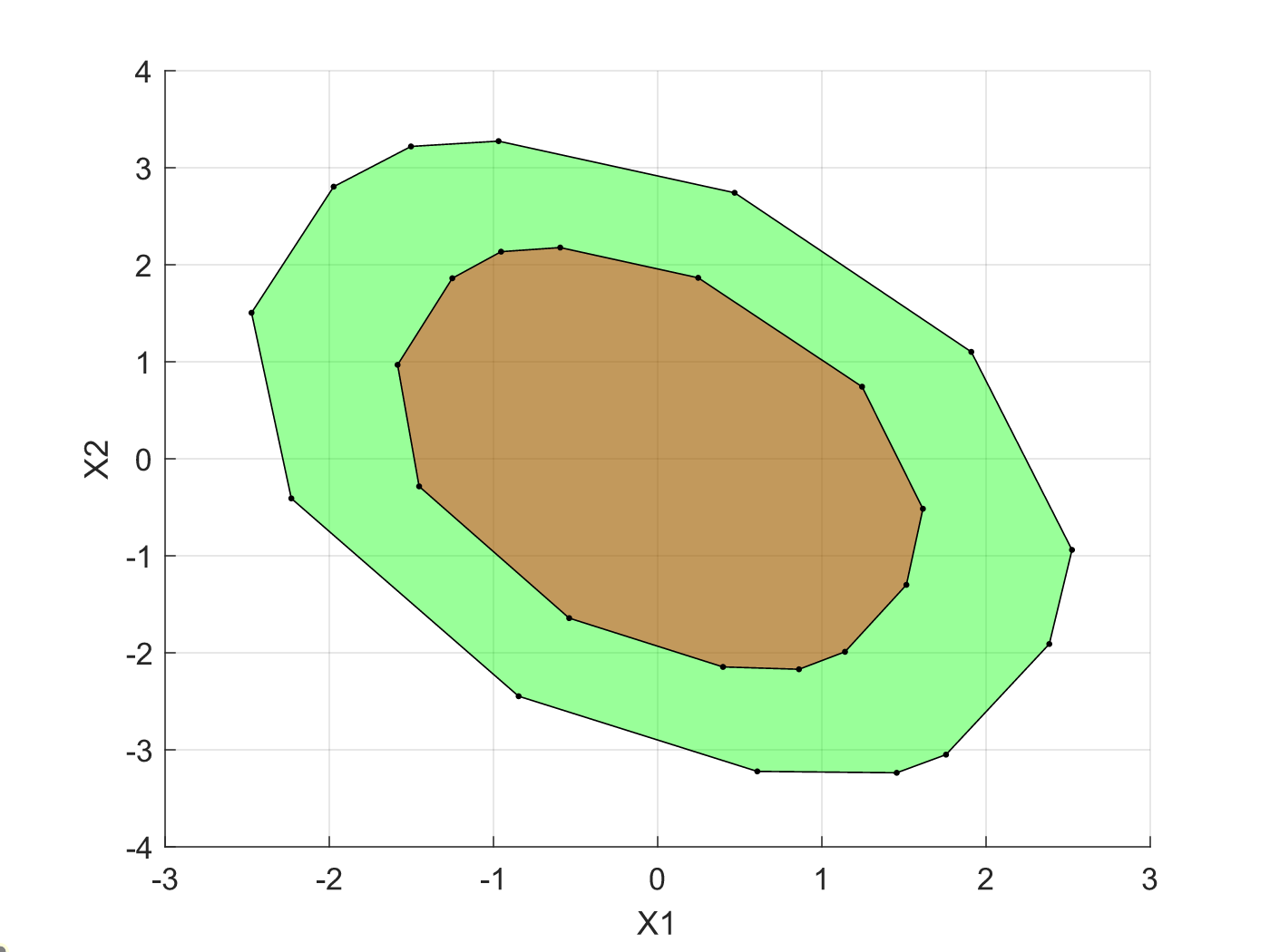}
    \caption{Reachable set growth for the damped oscillator with complex eigenvalues of $A$, $d = [1,0]^\top$. Red: nominal $\mathcal{R}_{B_{\mathrm{nom}}}$. Green: optimized $\mathcal{R}_{B^\ast}$. Horizontal axis: $x_1$. Vertical axis: $x_2$. Growth occurs in all directions, consistent with Section~\ref{complexeigdiscussion}.}
    \label{fig:imaginaryeigenvalues}
\end{figure}

\section{Summary and Conclusion}
\label{sec:Summaryandconclusion}
This paper addressed the problem of shaping the reachable set of a linear system by optimizing the input matrix $B$ within an admissible set $\mathcal{B}$. The objective was formalized through the Directional Growth Metric $G_d(B)$, which
quantifies the extent to which the reachable set is expanded or contracted along a prescribed direction of interest $d$. Theorem~\ref{thm:directional_optimal_B} established that, under the assumptions that $d$ is an eigenvector of $A^\top$ and that $A^\top$ has real eigenvalues, the joint optimization over $B$ and the control input reduces to $N$ linear subproblems, where $N$ is the number of vertices of the control polytope $\mathcal{U}$. The global maximizer $B^\ast$ of $G_d$ is then recovered by a single comparison across these $N$ candidates, as detailed in Algorithm~\ref{alg:iterative_B_update}. When these assumptions are relaxed, athough global optimality of $B^\ast$ is no longer guaranteed, but we heuristically show that the same method gives us good values of $B^\ast$. The algorithm was validated on two
systems: the ADMIRE linearized fighter jet model, demonstrating controlled
directional growth and shrinkage of the reachable set, and a damped oscillator with
complex eigenvalues, illustrating the multi-directional growth behavior predicted by
The relaxed-assumption analysis.
Several directions present themselves for future work. The most immediate extension
is to study the analogous problem of optimizing the system matrix $A$, and
subsequently the joint optimization of both $A$ and $B$, within their respective
admissible sets. Beyond linear systems, extending this framework to nonlinear dynamics
of the form $\dot{X} = f(X, u, B)$ presents a substantially richer challenge: unlike
the linear case, the co-state dynamics no longer admit a closed-form solution, and new
approaches (potentially drawing on numerical optimal control or learning-based
methods) will be required to make the optimization tractable.
\bibliography{mybib}{}
\bibliographystyle{unsrt}
\end{document}